\documentclass[11pt]{article}  

\usepackage{subfigure,enumerate,bm,amsmath,amsthm,wrapfig,array,color,algorithmic,amssymb}
\usepackage{graphics}
\usepackage{natbib}

\newtheorem{theorem}{\bf Theorem}[]
\newtheorem{lemma}[theorem]{Lemma}

\newtheorem{corollary}[theorem]{Corollary}

\newtheorem{thm}[theorem]{\bf Theorem}

\def\math#1{$#1$}

\def\mand#1{$$#1$$}
\def\v#1{{\mathbf #1}}
\def\frac#1#2{{#1\over #2}}

\def\eqan#1{\begin{eqnarray*}
#1
\end{eqnarray*}}

\DeclareSymbolFont{AMSb}{U}{msb}{m}{n}
\DeclareMathSymbol{\N}{\mathbin}{AMSb}{"4E}
\DeclareMathSymbol{\Z}{\mathbin}{AMSb}{"5A}
\DeclareMathSymbol{\R}{\mathbin}{AMSb}{"52}
\DeclareMathSymbol{\Q}{\mathbin}{AMSb}{"51}
\DeclareMathSymbol{\I}{\mathbin}{AMSb}{"49}
\DeclareMathSymbol{\C}{\mathbin}{AMSb}{"43}

\def\x{{\mathbf x}}

\def\z{{\mathbf z}}

\def\a{{\mathbf a}}

\def\norm#1{{\|#1\|}}

\def\dotfil{\leaders\hbox to 1.5mm{.}\hfill}
\newcounter{rmnum}
\def\RN#1{\setcounter{rmnum}{#1}\uppercase\expandafter{\romannumeral\value{rmnum}}}
\def\rn#1{\setcounter{rmnum}{#1}\expandafter{\romannumeral\value{rmnum}}}

\newcommand{\transp}{\ensuremath{^\text{\textsc{t}}}}

\newcommand{\mat}[1]{{\ensuremath{\textsc{#1}}}}

\def\v{{\mathbf v}}

\def\Atilde{\tilde\matA}

\def\Q{{\bm{Q}}}

\def\matA{\mat{A}}

\def\matX{\mat{X}}

\def\nnz{{\rm nnz}}

\DeclareMathSymbol{\Prob}{\mathbin}{AMSb}{"50}
\DeclareMathSymbol{\Exp}{\mathbin}{AMSb}{"45}
\newcommand\remove[1]{}

\begin{document}

\title{A Note On Estimating the Spectral Norm of A Matrix Efficiently}

\author{Malik Magdon-Ismail\\
CS Department, Rensselaer Polytechnic Institute,\\
Troy, NY 12180, USA. \\
{\sf magdon@cs.rpi.edu}
}

\maketitle
\begin{abstract}%
We give an efficient algorithm which can obtain a
relative error approximation to the spectral norm of a 
matrix, combining the power iteration method with some techniques from
matrix reconstruction which use random sampling.

\end{abstract}

\section{Introduction}
\label{section:intro}

For a matrix \math{\matA\in\R^{n\times d}}, \math{n\ge d}, we consider 
estimating its spectral norm 
\math{
\norm{\matA}=\max_{\norm{\x}=1}\norm{\matA\x}.
}
We give an algorithm to obtain a relative
error approximation to \math{\norm{\matA}} based on 
subsampling~\math{\matA} and then applying the power iteration. 
The algorithm is randomized, simple, and efficient. 
Slight improvements which give similar asymptotic running times 
could use a more sophisticated
method, e.g. a Lanczos method in lieu of the power method, however, 
we do not pursue that here. It is also known that no deterministic algorithm
can solve this problem \citep{kuczynski1992}, and so one must resort
to a randomized algorithm. 

\cite{oleary1979} showed good performance of the power method and  
\cite{kuczynski1992} gave a 
detailed analysis of the expected and high probability convergence of the
power method; \cite{woolfe2008} considered a randomized test for determining
if the spectral norm is above a value using multiple random starts.
We extend the results in \cite{kuczynski1992} to give a 
more efficient algorithm; we will give
a simplified, elementary proof of  the 
probabilistic convergence of the power method, a result 
asymptotically comparable to the one in \cite{kuczynski1992}; we will combine
this with a down-sampling of~\math{\matA} that preserves
the spectral norm to obtain  a randomized algorithm that realizes
the claim in Theorem~\ref{thm:1}. We quantify the running time in terms of 
\math{\nnz(\matA)} (the number of non-zero elements in~\math{\matA})
 and a parameter \math{\tau}, where
\mand{
\tau=O\left(\min\left\{
\frac{\nnz(\matA)}{\epsilon}\log \left(\frac{d}{\epsilon\delta}\right),
\frac{d^2}{\epsilon^3}\log^2 \left(\frac{d}{\epsilon\delta}\right)
%,
%\frac{\rho^2}{\epsilon^5}\log^3 \left(\frac{\rho}{\epsilon\delta}\right)
\right\}
\right).
}
and 
\math{\epsilon} is the relative error tolerance and
\math{\delta} is the failure probability.
\begin{thm}\label{thm:1}
Given \math{\matA\in\R^{n\times d}}, there is an algorithm which runs in 
\math{O(\nnz(\matA)+\tau)} and outputs an estimate 
\math{\tilde\sigma^2} which, with probability at least
\math{1-\delta}, satisfies
\mand{
(1-\epsilon)\norm{\matA}^2\le\tilde\sigma^2\le(1-\epsilon)\norm{\matA}^2.
}
\end{thm}
An estimate of the spectral norm can be used to efficiently 
compute the effective or numerical rank \math{\rho}
of \math{\matA}, 
\math{\rho=\sum_{ij}\matA_{ij}^2/\norm{\matA}^2}; \math{\rho} is 
useful in developinig efficient matrix algorithms, such as matrix 
multiplication \cite{zouzias2010, malik145}.
Notice that the running time is significantly faster than the 
\math{O(nd^2)} required to compute the spectral norm 
exactly via the singular value decomposition of \math{\matA}.
The algorithm, along with its proof of correctness is described in the
next section. 
The first term in \math{\tau} is implied by \cite{kuczynski1992}, so 
we focus on the second term.

\section{Estimating the Spectral Norm}
\label{section:sampling}

The algorithm has two basic steps. 
\begin{center}
\begin{algorithmic}[1]
\STATE Obtain a sketch \math{\Atilde} of \math{\matA} 
which has  smaller size  than \math{\matA} but for which
\math{\norm{\Atilde}\approx\norm{\matA}}.
\STATE Obtain
\math{\norm{\Atilde}} using the power iteration method.
\end{algorithmic}
\end{center}
For step 1, we  use an estimate proven in 
\cite{zouzias2010}, and independently in 
\cite{malik145}. 
Let \math{\matA=[\a_1,\ldots,\a_n]\transp}, where \math{\a_i\transp} are
the rows of \math{\matA}.
Define probabilities 
\mand{
p_i=\frac{\norm{a_i}^2}{\norm{\matA}_F^2},
}
where \math{\norm{\matA}_F^2=\sum_{i=1}^n\norm{a_i}^2} is the
Frobenius norm of \math{\matA}. Note that all \math{p_i} can be computed in
\math{O(\nnz(\matA))} time. 
Fix integer \math{r\ge 1};
we construct \math{\Atilde\in\R^{r\times d}} as follows.
Let \math{Z} be a vector valued random variable taking on the \math{n}
values
\math{\{\a_1/\sqrt{rp_1},\ldots,\a_n/\sqrt{rp_n}\}}, with
probabilities \math{\{p_1,\ldots,p_n\}}.
Let \math{\z_1,\ldots,\z_r} be \math{r} independent copies of
\math{Z}; the rows of \math{\Atilde} are the \math{\z_i},
\math{\Atilde=[\z_1,\ldots,\z_r]\transp}.
Note that, given the 
\math{p_i}, \math{\Atilde} can be obtained in additional time
\math{O(n+r\log r)} time.
\begin{lemma}[\cite{malik145}]
\label{lemma:lem1}
For \math{\epsilon>0}, if 
\math{r\ge(4d/\epsilon^2)\ln(2d/\delta)}, then
w.p. at least \math{1-\delta},
\mand{
\norm{\Atilde\transp\Atilde-\matA\transp\matA}\le
\epsilon\norm{\matA}^2.
}
\end{lemma}
\begin{corollary}
For \math{\epsilon>0}, if 
\math{r\ge(4d/\epsilon^2)\ln(2d/\delta)}, then
w.p. at least \math{1-\delta},
\mand{(1-\epsilon)\norm{\matA}^2\le\norm{\Atilde}^2
\le(1+\epsilon)\norm{\matA}^2.}
\end{corollary}
\begin{proof}
\eqan{
\norm{\Atilde\transp\Atilde}
&=&
\norm{\Atilde\transp\Atilde-\matA\transp\matA+\matA\transp\matA}
\ \le\ 
(1+\epsilon)\norm{\matA}^2;\\
\norm{\matA\transp\matA}
&=&
\norm{\matA\transp\matA-\Atilde\transp\Atilde+\Atilde\transp\Atilde}
\ \le\ 
\epsilon\norm{\matA}^2+\norm{\Atilde\transp\Atilde}.
}
\end{proof}
We have a sketch of \math{\matA} which preserves the
spectral norm; 
now, to obtain \math{\norm{\Atilde}^2}, we use  the power iteration.
Let \math{\matX\in\R^{r\times d}} be an arbitrary matrix, and 
\math{\x_0} a unit vector. 
For \math{n\ge1}, let 
\math{\x_n=\matX\transp\matX\x_{n-1}/\norm{\matX\transp\matX\x_{n-1}}}.
Note that multiplying by \math{\matX\transp\matX} can be done in 
\math{O(nd)} operations. Since \math{\x_n} is a 
unit vector,
\math{\norm{\matX\transp\matX\x_{n}}\le\norm{\matX}^2}.
Let \math{\x_0} be a random isotropic vector constructed using 
\math{d} independent standard Normal variates \math{z_1,\ldots,z_{d}};
so \math{\x_0\transp=[z_1,\ldots,z_{d}]/\sqrt{z_1^2+\cdots+z_{d_1}^2}}.
Let \math{\lambda_n^2=\norm{\matX\transp\matX\x_{n}}} be an estimate
for \math{\norm{\matX}^2} after \math{n} power iterations.
\begin{lemma}
\label{lemma:estimatenorm}
For \math{\epsilon>0} and a constant \math{c\le(\frac{2}{\pi}+2)^3},
 with probability at least \math{1-\delta},
\mand{\displaystyle
\lambda_n^2\ge
\frac{\norm{\matX}^2(1-\epsilon)}{\sqrt{1+\frac{cd}{\delta^3}\cdot 
(1-\epsilon)^{2(n+1)}}}.
} 
\end{lemma}
It immediately follows that for 
some constant \math{c}, if 
\math{n\ge (c/\epsilon) \log(d/\delta\epsilon)}, then
\math{\lambda_n^2\ge (1-\epsilon)\norm{\matX}^2}.
Since each power iteration takes \math{O(rd)} time, and we run 
\math{O((1/\epsilon)\log(d/\delta\epsilon))} power iterations, 
the running time is 
\math{O((rd/\epsilon)\log(d/\delta\epsilon))}.
Applying this to the estimate \math{\Atilde} from Lemma 
\ref{lemma:lem1}, 
with \math{r=(4d/\epsilon^2)\log(2d/\delta)}, and
we get Theorem~\ref{thm:1}.

\begin{proof}
Assume that \math{\x_0=\sum_{i=1}^{d_1}\alpha_i\v_i}, where \math{\v_i} are the
eigenvectors of \math{\matX\transp\matX} with corresponding 
eigenvalues
\math{\sigma_1^2\ge\cdots\ge\sigma_{d_1}^2}. 
Note, \math{\norm{\matX}^2=\sigma_1^2}.
If \math{\sigma_{d}^2\ge (1-\epsilon)\sigma_1^2}, then it
trivially follows that
\math{\norm{\matX\transp\matX\x_n}\ge (1-\epsilon)\sigma_1^2} 
for any \math{n}, so 
assume that 
\math{\sigma_{d}^2<(1-\epsilon)\sigma_1^2}.  We can thus partition the
singular values into those at least \math{(1-\epsilon)\sigma_1^2} 
and those which 
are smaller; the latter set is non-empty.
So assume for some \math{k<d}, \math{\sigma_k^2\ge(1-\epsilon)\sigma_1^2} and
\math{\sigma_{k+1}^2<(1-\epsilon)\sigma_1^2}. 
Since 
\mand{
\x_n=\frac{\sum_{i=1}^d\alpha_i\sigma_i^{2n}\v_i}{(\sum_{i=1}^d
\alpha_i^2\sigma_i^{4n})^{1/2}},
}
we therefore have:
\eqan{
\lambda_n^4
&=&
\norm{\matX\transp\matX\x_n}^2\\
&=&
\frac{\sum_{i=1}^{d}\alpha_i^2\sigma_i^{4(n+1)}}
{\sum_{i=1}^{d}\alpha_i^2\sigma_i^{4n}}\\
&\ge&
\frac{\sum_{i=1}^{k}\alpha_i^2\sigma_i^{4(n+1)}}
{\sum_{i=1}^{d}\alpha_i^2\sigma_i^{4n}}\\
&=&
\frac{\sum_{i=1}^{k}\alpha_i^2\sigma_i^{4(n+1)}}
{\sum_{i=1}^{k}\alpha_i^2\sigma_i^{4n}+\sum_{i=k+1}^{d}
\alpha_i^2\sigma_i^{4n}},\\
&=&
\sigma_1^4
\frac{\sum_{i=1}^{k}\alpha_i^2(\sigma_i/\sigma_1)^{4(n+1)}}
{\sum_{i=1}^{k}\alpha_i^2(\sigma_i/\sigma_1)^{4n}+\sum_{i=k+1}^{d}
\alpha_i^2(\sigma_i/\sigma_1)^{4n}},\\
&\buildrel(a)\over\ge&
\sigma_1^4
\frac{\sum_{i=1}^{k}\alpha_i^2(\sigma_i/\sigma_1)^{4(n+1)}}
{(1-\epsilon)^{-2}\sum_{i=1}^{k}\alpha_i^2(\sigma_i/\sigma_1)^{4(n+1)}+(1-\epsilon)^{-2n}},\\
&=&
\frac{\sigma_1^4}
{(1-\epsilon)^{-2}+(1-\epsilon)^{-2n}/\sum_{i=1}^{k}\alpha_i^2(\sigma_i/\sigma_1)^{4(n+1)}},\\
&\buildrel(b)\over\ge&
\frac{\sigma_1^4(1-\epsilon)^{2}}
{1+(1-\epsilon)^{-2(n+1)}/\alpha_1^2}.
}
(a) follows because for \math{i\ge k+1},
\math{\sigma_i^2<(1-\epsilon)\sigma_1^2}; for \math{i\le k},
\math{\sigma_1^2/\sigma_i^2\le (1-\epsilon)^{-2}}; 
and \math{\sum_{i\ge k+1}\alpha_i^2\le
\sum_{i\ge 1}\alpha_i^2=1}. (b) follows because 
\math{
\sum_{i=1}^{k}\alpha_i^2(\sigma_i/\sigma_1)^{4(n+1)}\ge \alpha_1^2}.
The theorem  now follows from the next lemma
by redefining  \math{\delta=(2/\pi+2)(\delta')^{1/3}}. 
\begin{lemma}\label{lemma:lemproof}
With probability at least \math{1-(2/\pi+2)(\delta')^{1/3}},
\math{\alpha_1^2\ge \delta'/d}.
\end{lemma}
To conclude the proof, we prove Lemma~\ref{lemma:lemproof}.
It is clear that \math{\Exp[\alpha_1^2]=1/d} from isotropy.
Without loss of generality, assume \math{\v_1} is aligned with the 
\math{z_1} axis. So 
\math{\alpha_1^2=z_1^2/\sum_iz_i^2}
(\math{z_1,\ldots,z_d} are independent standard normals).
For \math{\delta'<1},
we estimate \math{\Prob[\alpha_1^2\ge\delta'/d]} as follows:
\eqan{
\Prob\left[\alpha_1^2\ge\frac{\delta'}{d}\right]
&=&\Prob\left[\frac{z_1^2}{\sum_iz_i^2}\ge\frac{\delta'}{d}\right]\\
&=&\Prob\left[{z_1^2}\ge\frac{\delta'}{d}{\sum_{i\ge 1}z_i^2}\right]\\
&=&\Prob\left[{z_1^2}\ge\frac{\delta'}{d-\delta'}{\sum_{i\ge2}z_i^2}\right]\\
&\ge&\Prob\left[{z_1^2}\ge\frac{\delta'}{d-1}{\sum_{i\ge2}z_i^2}\right]\\
&\buildrel(a)\over=&\Prob\left[\chi^2_1\ge\frac{\delta'}{d-1}\chi^2_{d-1}\right],\\
&\buildrel(b)\over\ge&\Prob\left[\chi^2_1\ge \delta'+(\delta')^{2/3}\right]
\times\Prob\left[
\frac{\delta'}{d-1}\chi^2_{d-1}\le \delta'+(\delta')^{2/3}\right].\\
}
In (a) we compute the probability that a \math{\chi^2_1} random variable
exceeds a multiple of an independent \math{\chi^2_{d-1}} random variable, which
follows from the definition of the \math{\chi^2} distribution 
as a sum of squares
of 
independent standard normals.
(b) follows from independence and 
because one particular realization of the event in (a) is when
\math{\chi^2_1\ge \delta'+(\delta')^{2/3}} and 
\math{\delta'\chi^2_{d-1}/(d-1)\le \delta'+(\delta')^{2/3}}.
Since \math{\Exp[\chi^2_{d-1}/(d-1)]=1}, and 
\math{Var[\chi^2_{d-1}/(d-1)]=2/(d-1)}, by Chebyshev's inequality,
\mand{
\Prob\left[
\frac{\delta'}{d-1}\chi^2_{d-1}\le \delta'+(\delta')^{2/3}\right]
\ge 1-\frac{2(\delta')^{1/3}}{d-1}.
}
From the definition of the \math{\chi^2_1} distribution, we can bound
\math{\Prob[\chi^2_1\le\delta'+(\delta')^{2/3}]},
\mand{
\Prob[\chi^2_1\le\delta'+(\delta')^{2/3}]
=
\frac{1}{2^{1/2}\Gamma(1/2)}
\int_{0}^{\delta'+(\delta')^{2/3}}du\ u^{-1/2}e^{-u/2}
\le \sqrt{\frac{2}{\pi}}(\delta'+(\delta')^{2/3})^{1/2},
}
and so 
\mand{
\Prob\left[\alpha_1^2\ge\frac{\delta'}{d}\right]
\ge
\left(1-\sqrt{\frac{2}{\pi}}(\delta'+(\delta')^{2/3})^{1/2}\right)
\cdot
\left(1-\frac{2(\delta')^{1/3}}{d-1}\right)
\ge1-\left(\frac2\pi+2\right)(\delta')^{1/3}
.
}
%We are done because the median of \math{\chi^2_{d-1}} is less than 
%\math{d-1} \citep{sen1989}.
\end{proof}

%\subsubsection*{Acknowledgments}
{
\bibliographystyle{natbib}
\bibliography{active,mypapers,masterbib} 
}

%\appendix
%\input{proofs}

\end{document}